\documentclass[a4paper]{article}
\usepackage[margin=1.25in]{geometry}
\usepackage{amsmath,amsthm,amssymb,esint}
\usepackage{xparse}
\usepackage{graphicx}
\usepackage{braket}
\usepackage{amsthm}
\usepackage[usenames,dvipsnames,svgnames,table]{xcolor}
\usepackage[colorlinks = true, linkcolor = blue, urlcolor  = blue, citecolor = violet]{hyperref}

\newcommand{\fonci}{{F}_{i}}
\newcommand{\foncj}{{F}_{j}}
\newcommand{\fonc}{{F}}

\newcommand{\numset}[1]{\mathbf{#1}}
	\newcommand{\cc}{\numset{C}}
	\newcommand{\rr}{\numset{R}}
	
	\newcommand{\zz}{\numset{Z}}
	\newcommand{\nn}{\numset{N}}

	\newcommand{\one}{\bm{1}}
	
		\newcommand{\Exp}[1]{\mathrm{e}^{#1}}
	\newcommand{\ii}{\mathrm{i}}
	
	\renewcommand{\Re}{\operatorname{Re}}

	\makeatletter
	\providecommand*{\diff}%
		{\@ifnextchar^{\DIfF}{\DIfF^{}}}
	\def\DIfF^#1{%
		\mathop{\mathrm{\mathstrut d}}%
			\nolimits^{#1}\gobblespace}
	\def\gobblespace{%
		\futurelet\diffarg\opspace}
	\def\opspace{%
		\let\DiffSpace\!%
		\ifx\diffarg(%
			\let\DiffSpace\relax
		\else
			\ifx\diffarg[%
				\let\DiffSpace\relax
			\else
				\ifx\diffarg\{%
					\let\DiffSpace\relax
				\fi\fi\fi\DiffSpace}

	\providecommand*{\od}[3][]{%
		\frac{\diff^{#1}#2}{\diff #3^{#1}}}
	
	\renewcommand{\d}{\diff}

	\DeclareMathOperator{\spr}{spr}
	\DeclareMathOperator{\supp}{supp}

	\DeclareMathOperator{\tr}{tr}

	\newcommand{\h}{\mathfrak{h}}
	\newcommand{\B}{\mathcal{B}}
			\newcommand{\cB}{\mathcal{B}}
	
	\newcommand{\CAR}{\operatorname{CAR}}

	\newcommand{\Ga}[1][]{\Gamma^{-}_{#1}}
	
	\newcommand{\G}[1][]{\Gamma_{#1}\!}
	\newcommand{\dG}[1][]{\diff\Gamma_{#1}}

	\newcommand{\slim}{\operatorname{s-lim}\limits}

\newcommand{\sys}{\mathcal{S}}
\newcommand{\env}{\mathcal{E}}

\DeclareDocumentCommand \envstate { o } {%
  \IfNoValueTF {#1} {%
    \xi%
  }{%
    \xi_{#1}%
  }%
}
\DeclareDocumentCommand \sysstate { o } {%
  \IfNoValueTF {#1} {%
    \rho%
  }{%
    \rho_{#1}%
  }%
}

\renewcommand{\one}{\mathbf{1}}

\newcommand{\MM}{{M}}
\newcommand{\MMs}{{M}^*}

\newcommand{\ci}{\mathbf{S}^1}
\renewcommand{\h}{\mathcal{H}}

\renewcommand{\Ga}[1][]{\Gamma^{-}_{#1}}

\renewcommand{\G}[1][]{\Gamma_{#1}}
\renewcommand{\dG}[1][]{\diff\Gamma_{#1}}

\newtheorem{corollary}{Corollary}[section]
\newtheorem{proposition}[corollary]{Proposition}
\newtheorem{lemma}[corollary]{Lemma}
\newtheorem{theorem}[corollary]{Theorem}

\newtheorem{remark}[corollary]{Remark}
\newtheorem{example}[corollary]{Example}
\begin{document}

\title{On Fermionic walkers interacting with a correlated structured environment}


\author{Renaud Raqu\'epas}


\date{}

\maketitle
{
\begin{center}\small
	\begin{tabular}{c c c c}
		Univ.\ Grenoble Alpes & &&  McGill University\\
		CNRS, Institut Fourier & && Dept.\ of Mathematics and Statistics\\
		F-38\,000 Grenoble  & && 1005--805 rue Sherbrooke Ouest \\
		France & &&  Montr\'eal (Qu\'ebec){\,\ }H3A 0B9 \\
	\end{tabular}
\end{center}
}
\smallskip

\begin{abstract}
  We study the large-time behaviour of a sample~$\sys$ consisting of an ensemble of fermionic walkers on a graph interacting with a structured infinite reservoir of fermions~$\env$ through an exchange of particles in preferred states. We describe the asymptotic state of~$\sys$ in terms the initial state of~$\env$, with especially simple formulae in the limit of small coupling strength. We also study the particle fluxes into the different parts of the reservoir.

\end{abstract}

\section{Introduction}

At times arising as successful approximations of continuous-time quantum evolutions that are of interest to experimental and theoretical physicists, at times arising as natural quantum counterparts of classical discrete-time processes in different sciences, discrete-time quantum evolutions have become an extensively studied topics in mathematical physics. Examples of the first kind appear in the effective description of quantum systems that repeatedly interact with probes~\cite{KM00,BJM14}, quantum systems that undergo time-periodic driving \cite{Ho74,Ya77} or models related to the quantum Hall effect, where strong perpendicular magnetic fields are involved~\cite{CC88,KOK05}; examples of the second kind can be found in quantum information science~\cite{Me96,Wa01}.

Discrete-time quantum evolutions where the dynamics for a single time step only couples {neighbouring} sites of a (possibly infinite) graph are often referred to as \textit{quantum walks} and have been studied extensively in the last twenty years; see for example~\cite{AAKV01,Ke03,Ve12,Por}. Most works on the subject consider a single particle, called a \emph{quantum walker}. However, interesting phenomena arise when two walkers are coupled~\cite{AA+12,SB+17} and natural questions concerning the collective behaviour of a variable number of walkers arise, especially by analogy with phenomena of Hamiltonian quantum statistical mechanics such as return to equilibrium, existence of nonequilibrium steady states, entropy production, {et cetera}. This was initiated by Hamza and Joye in the work~\cite{HJ17}, where they prove a form of return to equilibrium for ensembles of walkers interacting with a chain of auxiliary fermions in two special cases of the model considered here.

The main model under consideration concerns a finite graph on which a variable number of noninteracting fermions may hop to a neighbouring vertex at discrete times. We call this component of the system the sample~$\sys$. The free dynamics in Fock space is given by the second quantization of a one-particle unitary matrix~$W$.
To model the interaction with an environment, we introduce an auxiliary bi-infinite chain~$\env$ of sites where fermions are forced to hop to their left at discrete times. The free dynamics there is thus described by the second quantization of a shift operator~$S$.
These two components of the system then interact through a term which allows the exchange of particles in preferred states. The intensity of this exchange is controlled by a coupling constant~$\alpha$.
They undergo a step of free evolution, a step of interaction evolution, a step of free evolution, and so on.

We consider the case in which the environment~$\env$ is in an initial state which is translational-invariant, gauge-invariant (\textsc{gi}) and quasi-free (\textsc{qf}), described at the one-particle level by a function~$F$ of the shift operator~$S$.
Most notably, we show that the asymptotic state in the sample~$\sys$ is then also \textsc{giqf} and completely described by its own dynamics~$W$ and the same function~$F$, up to an interaction-dependent deformation which vanishes in the limit~$\alpha \to 0$. This deformation is explicit.
This is a generalization of the results of Hamza and Joye in~\cite{HJ17}: they had covered the cases where the sites in the environment~$\env$ are uncorrelated or where the free hopping in the sample~$\sys$ is a uni-directional shift on a ring.

We further generalize to more complicated structured environments by considering~$m \geq 1$ internal degrees of freedom at each site of the bi-infinite chain and a dynamics determined at the one-particle level by $S \otimes U$ for some unitary $m$-by-$m$ matrix~$U$ with simple eigenvalues.
Invariance of the state under the dynamics and under translations along the chain yield a set $(\pi_i)_{i=1}^m$ of rank-one projections such that the symbol describing the initial state of the enrivonment~$\env$ is $\sum_{i=1}^{m}(2\Re\fonci(S^*\otimes U^*))(\one \otimes \pi_i)$ for some family $(\fonci)_{i=1}^m$ of functions.
In this generalized version of the model, the asymptotic state in the sample~$\sys$ is again \textsc{giqf} and completely described by its own dynamics~$W$ and the functions~$(\fonci)_{i=1}^m$\,---\,again up to an explicit $\alpha$-dependent correction\,---\,; see Theorem~\ref{thm:q-f}. We can also study the flux of particles into the different parts~$(\env_i)_{i=1}^m$ of the environment~$\env$ and determine\,---\,for small enough couplings\,---\,the signs by comparing the values of the functions~$(\fonci)_{i=1}^{m}$ at the point~$1$.

In Section~\ref{sec:sys-dyn}, we introduce more precisely the description of the spaces, observables and dynamics; we leave some comments on the choice of quantum statistics and interaction for Appendix~\ref{sec:stats}.
We devote Section~\ref{sec:state} to a discussion of the initial state of the system. This is where we introduce the decomposition of the translation-invariant  environment~$\env$ into different parts~$(\env_i)_{i=1}^m$, to each of which is associated a scalar function~$\fonci$.
In Section~\ref{sec:asymptotics}, we state our main results on the asymptotic state in the sample and asymptotic fluxes out of the different parts of the environment.
We present more concrete examples with $m=1$ in Section~\ref{sec:ex}.
For a model on a ring, we get an explicit expression for the profile of the particle density as a function of the node of the graph and for the corresponding correlations between the occupation of two nodes. We also study the role of disorder and of the essential range of the function~$\fonc$ in the limit of an infinitely large sample.

\paragraph*{Acknowledgements}

The author would like to thank Alain Joye for introduction to these questions. The author acknowledges financial support from the Natural Sciences and Engineering Research Council of Canada (NSERC), and from the French Agence Nationale de la Recherche through grant ANR-17-CE40-0006.

\section{The system and its dynamics} \label{sec:sys-dyn}

\subsection{A quantum walk in a sample}

In this subsection, we introduce the description of the sample, which is the small system of interest interacting with an environment. We start by considering a single walker for two classes of graphs, and then introduce our general assumptions and the passage to a variable number of fermionic walkers.

\subsubsection{The coined walk on a cycle}

We wish to start by describing the motion of a spin-$\tfrac 12$ quantum walker on a cycle of~$n$ vertices. Let $\{\delta_0, \delta_1, \dotsc, \delta_{n-1}\}$ be an orthonormal basis for~$\cc^n$, the Hilbert space for the position, and let $\{e_{-1}, e_{+1}\}$ be an orthonormal basis for~$\cc^2$, the Hilbert space for the spin. The unitary
\[
	W_1 := \sum_{\nu=0}^{n-1} \sum_{\tau = \pm 1} \delta_{\nu + \tau} \otimes e_{\tau} \braket{\delta_\nu \otimes e_\tau  ,{\cdot\,}},
\]
on $\cc^n \otimes \cc^2$, with the arithmetics on the index~$\nu$ understood modulo~$n$, can be interpreted as follows. If a walker is located on the site~$\nu$ with its spin up [resp. down], it moves clockwise [resp. counter-clockwise] and its spin variable is unchanged. A second unitary
\[
	W_2 := \sum_{\nu = 0}^{n-1} \delta_\nu \braket{\delta_\nu,{\cdot\,}} \otimes C_\nu,
\]
where each $C_\nu$ is a $2$-by-$2$ unitary called a coin matrix, is used to locally change the spin variable. The so-called Hadamard coin~$\tfrac{1}{\sqrt 2}(\begin{smallmatrix} 1 & 1 \\ -1 & 1 \end{smallmatrix})$ is most often considered in the literature.

The product
\[
  W := W_1 W_2,
\]
a unitary on the finite dimensional Hilbert space $\h_\sys = \cc^n \otimes \cc^2$, gives a dynamics for a single time step where a quantum coin is shuffled, and the walker hops to a neighbouring site or the other depending on the outcome of the coin.

If each coin matrix~$C_\nu$ is such that $\braket{e_{+1}, C_\nu e_{+1}}\braket{e_{+1},C_\nu e_{-1}} \neq 0$, then the vector $\psi_* = \delta_0 \otimes e_{-1}$ is cyclic for~$W$. Such a preferred state will be used later to introduce exchange of particles with the environment, and cyclicity will ensure good propagation of the interaction.

\subsubsection{Coined walks on some more general graphs}

We now treat a second class of graphs, called class-1 regular graphs in~\cite{Por}, for which we can give a precise description of the dynamics without introducing too much machinery from graph theory.

Let~$G$ be a finite $r$-regular graph with~$n$ vertices whose edges can be coloured with~$r$ colours, $a = 0, \dotsc, r-1$, and consider the unitary
\[
	W_1 := \sum_{\nu = 0}^{n-1} \sum_{a = 0}^{r-1} \delta_{\nu'(\nu,a)} \otimes e_a \braket{\delta_{\nu} \otimes  e_a,{\cdot\,}}
\]
on $\cc^n\otimes\cc^r$, where $\nu'(\nu,a)$ is the unique vertex such that~$(\nu,\nu')$ is an edge of colour~$a$. By a standard result, the colouring property forces~$n$ to be even. The action of~$W_1$ is interpreted as follows. If a walker is at the node~$\nu$, looking along the edge of color~$a$, after one time step, it will go to the node~$\nu'$ at the other end of this edge, now facing where it came from.

We also introduce a unitary coin matrix~$C_\nu$ at each node~$\nu$ and the unitary
\[
	W_2 := \sum_{\nu = 0}^{n-1} \delta_\nu \braket{\delta_\nu,{\cdot\,}} \otimes C_\nu.
\]
on $\cc^n\otimes\cc^r$. The role of the coin matrix is to change the direction in which the walker is looking into a normalized superposition of directions.

Finally, we set
\[
	W := W_1 W_2.
\]
This unitary matrix~$W$ on the finite-dimensional Hilbert space~$\h_\sys = \cc^n \otimes \cc^r$ of dimension~$d = nr$ describes the discrete-time dynamics for a single walker on the graph~$G$. Under generic assumptions on the coin matrices, the matrix~$W$ will admit a cyclic vector~$\psi_*$.

\subsubsection{General formulation and second quantization}

Quantum walks on more general classes of finite graphs can be given a similar description using slightly heavier machinery from graph theory; see for example Chapter~7 in~\cite{Por}. Any description of a quantum walk on a (directed) graph~$G = (V,E)$ should involve a unitary~$W$ on a space of the form $\cc^{|V|} \otimes \cc^r$ whose matrix elements $\braket{\delta_{\nu'} \otimes w, W(\delta_\nu \otimes w)}$ vanish unless~$(\nu,\nu')$ is in the set~$E$ of (directed) edges of~$G$.

The specific structure of~$W$ inherited from the graph is irrelevant for most of our computations and results and will only be used in the interpretation of some asymptotic results, especially in Section~\ref{sec:ex}. Hence, for the rest of the paper we will consider that the one-particle dynamics~$W$ on a $d$-dimensional Hilbert space~$\h_\sys$ and the unit vector~$\psi_* \in \h_\sys$ are admissible whenever~$W$ is unitary and~$\psi_*$ is cyclic for~$W$.

A variable number of non-interacting fermionic walkers are then described on the algebra~$\CAR(\h_\sys)$ of canonical anti-commutation relations, which we represent on the fermionic Fock space~$\Ga(\h_\sys)$, with the creation [resp. annihilation] operator associated to the one-particle state~$\psi \in \h_\sys$ denoted by~$a^*(\psi)$ [resp.~$a(\psi)$]. We have the usual norm identity $\|a^\sharp(\psi)\| = \|\psi\|$ for each $\psi \in \h_\sys$. Here and in what follows,~$a^\sharp$ is used as a placeholder for either~$a^*$ or~$a$.

The discrete-time dynamics there is given by
\[
	a^\sharp(\psi) \mapsto a^\sharp(W^*\psi)
\]
for all~$\psi \in \h_\sys$, extended by linearity to~$\CAR(\h_\sys)$.  We refer to this as the free dynamics in the \emph{sample}~$\sys$.

\subsection{Interaction with an environment}

To model interaction of the sample~$\sys$ of interest with an environment~$\env$, we introduce the Hilbert space~$\h_\env := \ell^2(\zz) \otimes \cc^m$ and the one-particle dynamics $S \otimes U$ there, where~$S$ is the shift $\delta_{\ell} \mapsto \delta_{\ell - 1}$ on~$\ell^2(\zz) = \{\sum_{\ell \in \zz} a_\ell \delta_\ell : \{a_\ell : \ell \in \zz\}\subset \cc, \sum_{\ell \in \zz} |a_\ell|^2 < \infty\}$ and~$U$ is an arbitrary unitary on~$\cc^m$ with simple eigenvalues.

A variable number of walkers in this environment are then described on the algebra~$\CAR(\h_\env)$ with discrete-time dynamics given by
\[
	b^\sharp(\varphi) \mapsto b^\sharp((S^* \otimes U^*)\varphi)
\]
for all~$\varphi \in \h_\env$, extended by linearity to~$\CAR(\h_\env)$.\footnote{We represent the CAR in the environment~$\env$ on~$\Ga(\h_\env)$ and use~$b^*(\varphi)$ [resp.~$b(\varphi)$] for the creation [resp.~annihilation] operator associated to the vector~$\varphi \in \h_\env$. We use $b^\sharp$ as a placeholder for either~$b^*$ or~$b$.}
We refer to this as the free dynamics in the \emph{environment}~$\env$.

Recall that we assume that the unitary~$W$ used to describe the free dynamics in the sample~$\sys$ admits a cyclic vector $\psi_* \in \h_\sys$. We use this vector to construct the interaction unitary
\[
		K_\alpha := \exp[-\ii\alpha(b^*(\delta_0\otimes v) \otimes a(\psi_*) + b(\delta_0\otimes v) \otimes a^*(\psi_*))]
\]
on~$\CAR(\h_\env) \otimes \CAR(\h_\sys)$, where $v \in \cc^m$ is a unit vector and where $\alpha \in \rr \setminus \pi\zz$ is a coupling constant. By considering this interaction on the algebra~$\CAR(\h_\env) \otimes \CAR(\h_\sys)$, we are specifying that the fermions in the sample and in the environment are distinguishable. Folklore suggests that such details of the description of the environment should not influence the asymptotics in the sample~$\sys$; see Appendix~\ref{sec:stats} for a further discussion. We will often omit tensored identities and tensor products between $a^\sharp$s and $b^\sharp$s.

To the unitary~$K_\alpha$ is closely related the rank-one operator
\begin{align*}
	\iota : \h_\sys &\to \h_\env \\
		\psi &\mapsto \delta_0 \otimes v\braket{\psi_*,\psi}
\end{align*}
and its adjoint~$\iota^*$. We introduce the shorthands~$P$ for the projector~$\iota^*\iota$ on~$\h_\sys$ and~$Q$ for the projector~$\iota\iota^*$ on~$\h_\env$.

We consider the following coupled discrete-time Heisenberg dynamics:
\begin{align*}
	\tau(X) :=
		 K_\alpha^* (\Gamma(S^* \otimes U^*) \otimes \Gamma(W^*)) X (\Gamma(S \otimes U) \otimes \Gamma(W)) K_\alpha
\end{align*}
for all observables~$X$ in $\CAR(\h_\env) \otimes \CAR(\h_\sys)$. This describes a dynamics where both the sample and the environment first evolve independently, and then are allowed to interact through a term which exchanges particles in the states $\psi_* \in\h_\sys$ and $\delta_0 \otimes v \in \h_\env$. This dynamics preserves the total number of fermions and the strength of the interaction is controlled by the coupling constant~$\alpha$.

	The joint dynamics is given by a product of unitaries which cannot be written as the exponential of a a sum of physically significant free Hamiltonians and an interaction potential. For this reason, many of the tools and definitions from Hamiltonian quantum statistical mechanics\,---\,most notably the notion of temperature\,---\,are not available.

We have the following lemmas, which will serve as building blocks for our computation of the evolution of observables of interest. As can be seen from the formulae, the case $\alpha \in \pi\zz$ is in some sense trivial and we later exclude it from our analysis.

\begin{lemma}
	For all $\psi \in \h_\sys$
	\begin{equation}\label{eq:one-a-KK}
		K_\alpha a^\sharp(\psi) K_\alpha = a^{\sharp}((\one + (\cos\alpha-1)P)\psi) - \ii\sin\alpha\,b^\sharp(\iota \psi),
	\end{equation}
	and for all $\varphi \in \h_\env$
	\begin{equation}\label{eq:one-b-KK}
		K_\alpha b^\sharp(\varphi) K_\alpha = b^{\sharp}((\one + (\cos\alpha-1)Q)\varphi) - \ii\sin\alpha\,a^\sharp(\iota^* \varphi).
	\end{equation}
\end{lemma}

\begin{proof}
	The functions $\alpha \mapsto K_\alpha b^\sharp(\varphi) K_\alpha$ and $\alpha \mapsto K_\alpha a^\sharp(\psi) K_\alpha$ are real-analytic and hence completely determined by their derivatives in~$\alpha = 0$. Both formulae in the lemma are obtained by differentiating and exploiting anticommutation relations. We omit the subscript~$\alpha$ to lighten the notation.

	Indeed, differentiating,
	\begin{align*}
			\od{}{\alpha} K a^*(\psi) K &= -\ii K \{ b(\delta_0 \otimes v)a^*(\psi_*) +  b^*(\delta_0\otimes v)a(\psi_*),a^*(\psi)\} K \\
				&= -\ii K \braket{\psi_*,\psi}b^*(\delta_0 \otimes v) K \\
				&= -\ii K b^*(\iota \psi) K.
	\end{align*}
	Similarly,
	\[
		\od{}{\alpha} K b^*(\varphi) K = -\ii K a^*(\iota^* \varphi) K.
	\]
	As a consequence of the definition of $P = \iota^*\iota$ and the identity $\iota \iota^*\iota = \iota$, for all $n\geq 0$,
	\begin{align*}
	 	\od[2n+1]{}{\alpha} K a^*(\psi) K &= (-\ii)^{2n+1} K b^*(\iota\psi) K, \\
	 	\od[2n+2]{}{\alpha} K a^*(\psi) K &= (-\ii)^{2n+2} K a^*(P\psi) K
	 \end{align*}
	so that
	\[
		K a^*(\psi) K = a^*(\psi) + \sum_{n \geq 0} \frac{(-\ii\alpha)^{2n+1}}{(2n+1)!} b^*(\iota\psi) + \frac{(-\ii\alpha)^{2n+2}}{(2n+2)!} a^*(P\psi).
	\]
	The other formula is obtained similarly.
\end{proof}

\begin{lemma}\label{lem:one-tau}
	For all $\psi_1, \psi_2 \in \h_\sys$
	\begin{equation}\label{eq:one-aa-tau}
	\begin{split}
		\tau(a^{\sharp_1}(\psi_1)a^{\sharp_2}(\psi_2))
      & = a^{\sharp_1}((\one + (\cos\alpha-1)P)W^*\psi_1)a^{\sharp_2}((\one + (\cos\alpha-1)P)W^*\psi_2) \\
			&\qquad + \ii\sin\alpha\,a^{\sharp_1}((\one + (\cos\alpha-1)P)W^*\psi_1) b^{\sharp_2}(\iota W^*\psi_2)\\
			&\qquad - \ii\sin\alpha\,b^{\sharp_1}(\iota W^*\psi_1)a^{\sharp_2}((\one + (\cos\alpha-1)P)W^*\psi_2) \\
			&\qquad + \sin^2\alpha\,b^{\sharp_1}(\iota W^*\psi_1)b^{\sharp_2}(\iota W^*\psi_2)
	\end{split}
	\end{equation}
	and, for all $\varphi_1,\varphi_2 \in \h_\env$,
	\begin{equation*}
	\begin{split}
		\tau(b^{\sharp_1}(\varphi_1)b^{\sharp_2}(\varphi_2))
      & = b^{\sharp_1}((\one + (\cos\alpha-1)Q)S^*U^*\varphi_1)b^{\sharp_2}((\one + (\cos\alpha-1)Q)S^*U^*\varphi_2) \\
			&\qquad -\ii\sin\alpha\,a^{\sharp_1}(\iota^* S^*U^*\varphi_1)b^{\sharp_2}((\one + (\cos\alpha-1)Q)S^*U^*\varphi_2) \\
			&\qquad +\ii\sin\alpha\,b^{\sharp_1}((\one + (\cos\alpha-1)Q)S^*U^*\varphi_1)a^{\sharp_2}(\iota^* S^*U^*\varphi_2) \\
			&\qquad + \sin^2\alpha\,a^{\sharp_1}(\iota^* S^*U^*\varphi_1)a^{\sharp_2}(\iota^* S^*U^*\varphi_2).
	\end{split}
\end{equation*}
	Moreover, for all $\psi_1 \in \h_\sys$ and $\varphi_2 \in \h_\env$,
	\begin{equation*}
	\begin{split}
		\tau(b^{\sharp_1}(\varphi_1)a^{\sharp_2}(\psi_2))
      & = b^{\sharp_1}((\one + (\cos\alpha-1)Q)S^*U^*\varphi_1)a^{\sharp_2}((\one + (\cos\alpha-1)P)W^*\psi_2) \\
			&\qquad +\ii\sin\alpha\,a^{\sharp_1}(\iota^* S^*U^*\varphi_1)a^{\sharp_2}((\one + (\cos\alpha-1)P)W^*\psi_2) \\
			&\qquad - \ii\sin\alpha\,b^{\sharp_1}((\one + (\cos\alpha-1)Q)S^*U^*\varphi_1)b^{\sharp_2}(\iota W^*\psi_2) \\
			&\qquad + \sin^2\alpha\,a^{\sharp_1}(\iota^* S^*U^*\varphi_1) b^{\sharp_2}(\iota W^*\psi_2).
	\end{split}
\end{equation*}
\end{lemma}

\begin{proof}
	Because
	\[
		(\Gamma(S^* \otimes U^*) \otimes \Gamma(W^*)) a^\sharp(\psi) (\Gamma(S \otimes U) \otimes \Gamma(W)) = a^\sharp(W^*\psi)
	\]
	and
	\[
		(\Gamma(S^* \otimes U^*) \otimes \Gamma(W^*)) b^\sharp(\varphi) (\Gamma(S \otimes U) \otimes \Gamma(W)) = b^\sharp((S^* \otimes U^*)\varphi),
	\]
	the formulae follow from the previous lemma and the fact that $\tau$ is a morphism.
\end{proof}

\section{The initial state} \label{sec:state}

The initial state of the compound system is taken to be a tensor product of an arbitrary even state~$\rho$ on~$\CAR(\h_\sys)$ with a \textsc{giqf} state~$\omega_\Sigma$ on~$\CAR(\h_\env)$ with symbol~$\Sigma \in \cB(\h_\env)$.
In other words, the state $\rho$ vanishes on all monomials of odd degree, $$\rho(a^{\sharp_1}(\psi_1)a^{\sharp_2}(\psi_2) \dotsb a^{\sharp_{2N+1}}(\psi_{2N+1})) = 0,$$
and there exists~$\Sigma \in \cB(\h_\env)$, called the {\it density} or {\it symbol} of the state $\omega_\Sigma$, such that
\begin{align*}
		\omega_\Sigma(b^{*}(\varphi_1) \dotsb b^{*}(\varphi_N)b(\varphi'_{N'}) \dotsb b(\varphi'_1)) =
			 \delta_{NN'} \det [\braket{\varphi_{\mu},\Sigma \varphi'_{\mu'}}]_{\mu,\mu'=1}^N.
\end{align*}
In addition to the usual requirement that $0 \leq \Sigma \leq \one$, we require that the symbol is translation invariant, in the sense that
\begin{equation}\label{eq:ti}
	[\Sigma, S \otimes \one] = 0,
\end{equation}
and invariant under the one-particle free dynamics in the environment, {i.e.}
\begin{equation}\label{eq:invar}
	[\Sigma, S \otimes U] = 0.
\end{equation}
Finally, we assume that
\begin{equation}\label{eq:decay}
	\sum_{\ell \in \zz} |\braket{\delta_{-\ell} \otimes w, \Sigma(\delta_0 \otimes w')}| < \infty
\end{equation}
for all $w, w' \in \cc^m$. This is a technical assumption that will ensure absolute convergence of important series. It can be interpreted as a decay of spatial correlations in the environment.

\begin{remark}
  The case $m=1$ (and $U = 1$) is the one considered in~\cite{HJ17}. They proved a special case of Theorem~\ref{thm:q-f} below in two situations: when the symbol~$\Sigma$ is proportional to the identity or when~$W$ is a shift on~$\h_\sys$.

  Allowing~$\Sigma$ to have off-diagonal terms allows us to study the effect of correlations in the structured environment~$\env$ (absent in repeated interaction systems) on the asymptotics of the sample~$\sys$, and considering~$m>1$ is a first step towards studying non-equilibrium situations.
\end{remark}

Recall that we assumed~$U$ to have simple eigenvalues: its spectral projectors~$\pi_i$, $i=1,\dotsc,m$, have rank one and each of them corresponds to a unit eigenvector~$x_i$ for a distinct eigenvalue. Then, for any indices $k \in \zz$ and~$i \in \{1,\dotsc,m\}$, we have
\begin{align*}
	\Sigma (\delta_k \otimes x_i)
		&= \sum_{\ell \in \zz} \sum_{j=1}^m \delta_{\ell+k} \otimes x_j \braket{\delta_{\ell+k} \otimes {x_j},\Sigma (\delta_k \otimes x_i)} \\
		&= \sum_{\ell \in \zz} \delta_{\ell+k} \otimes x_i \braket{\delta_{\ell+k} \otimes {x_i},\Sigma (\delta_k \otimes x_i)} \\
		&= \sum_{\ell \in \zz} \braket{\delta_0 \otimes x_i, \Sigma(S \otimes U)^\ell(\delta_0 \otimes x_i)}(S^* \otimes U^*)^\ell (\delta_k \otimes x_i) \\
		&= 2 \Re \Big(\frac 12 \braket{\delta_0 \otimes x_i, \Sigma(\delta_0 \otimes x_i)}
			\\ &\qquad\qquad{}+
			\sum_{\ell = 1}^\infty \braket{\delta_0 \otimes x_i, \Sigma(S \otimes U)^\ell(\delta_0 \otimes x_i)}(S^* \otimes U^*)^\ell  \Big)(\delta_k \otimes x_i).
\end{align*}
The coefficients in the sum on the right-hand side define analytic functions in the unit disk,
\begin{equation}\label{eq:sigma-plus}
	\begin{split}
	\fonci
	(\zeta) :=
		\frac 12 \braket{\delta_0 \otimes x_i, \Sigma(\delta_0 \otimes x_i)}
			+\sum_{\ell = 1}^\infty \braket{\delta_0 \otimes x_i, \Sigma(S \otimes U)^\ell(\delta_0 \otimes x_i)}\zeta^\ell,
	\end{split}
\end{equation}
which extend continuously up to the boundary. In particular,
the real number~
$2\fonci(0) = \braket{\delta_0 \otimes x_i, \Sigma(\delta_0 \otimes x_i)}$ is the translation-invariant particle density in the $i$th subspace of the reservoir.

Note that the one-particle operator~$\Sigma$ specifying the initial state of the environment~$\env$ is a sum of functions of the dynamics $S^* \otimes U^*$:
\begin{equation}\label{eq:sigma-series}
	\Sigma = \sum_{i=1}^{m}(2\Re\fonci(S^*\otimes U^*))(\one \otimes \pi_i).
\end{equation}

One of the main results of this paper is that the asymptotic state in the sample is a \textsc{giqf} state whose density can be expressed as a sum of those same functions, but now of a matrix describing the dynamics in the sample. As discussed in Appendix~\ref{sec:stats}, this state remains the same when one considers some variants of this model.

\begin{remark}\label{rem:int-sigma}
	The meaning of $\fonci(B)$ for a linear operator~$B$ with $\|B\| \leq 1$ is given through the power series~\eqref{eq:sigma-plus}. Convergence in norm holds by our assumption~\eqref{eq:decay} on the decay of correlations. The real part of $\fonci(B)$ is simply the self-adjoint operator $\tfrac 12 (\fonci(B) + (\fonci(B))^*)$.

	 If the function extends analytically to a disk of radius $r > \spr(B)$, then we have the integral formula
	\begin{equation}\label{eq:int-sigma}
		\fonci(B)
			= \frac{1}{2\pi\ii} \ointctrclockwise_{r \ci} \fonci(\zeta) R(\zeta,B)  \d\zeta
	\end{equation}
	by the residue theorem and the Neumann series
	 \(R(\zeta,B) = \frac 1\zeta \sum_{\ell = 0}^\infty \zeta^{-\ell} B^{\ell} \)
	for the resolvent of~$B$ for $|\zeta| > \spr B$.
\end{remark}

\section{Asymptotics}\label{sec:asymptotics}

We are interested in the behaviour of the system in the limit $t \to \infty$ along~$\nn$. In particular, we study the behaviour of the state in the sample~$\sys$ and flux observables.

Throughout the section, we will use the following hypotheses:
with $W$ a unitary on~$\cc^d$, $\Sigma$ a non-negative definite bounded operator on~$\ell^2(\zz) \otimes \cc^m$,~$U$ a simple unitary on~$\cc^m$ and~$\alpha$ the coupling constant for the interaction,
\begin{enumerate}
	\item[(i)] the initial state for the compound system is the tensor product of the \textsc{giqf} state $\omega_\Sigma$ on~$\CAR(\ell^2(\zz) \otimes \cc^m)$, where~$\Sigma$ satisfies $[\Sigma,S \otimes U] = [\Sigma, S \otimes \one] = 0$, with an even state~$\rho$ on~$\CAR(\cc^d)$;
	\item[(ii)] the decay condition~\eqref{eq:decay} on~$\Sigma$ holds;
	\item[(iii)] the vector $\psi_*$ used to define the interaction is cyclic for~$W$;
	\item[(iv)] the coupling constant~$\alpha$ is not an integer multiple of~$\pi$.
\end{enumerate}

\subsection{The state in the sample}

In what follows, we use
\[
  \MM := W(\one + (\cos\alpha - 1)P)
\]
for the matrix appearing in the formulae from Section~\ref{sec:sys-dyn} for time-evolved creation and annihilation operators. The dependence on the coupling constant~$\alpha$ is not apparent in this notation but should be kept in mind. Consequently, we will write statements such as $\MM \to W$ as $\alpha \to 0$, referring to the initial dependence of~$\MM$ on~$\alpha$.

Since~$W$ is a unitary and~$P = \iota^*\iota$ is an orthogonal projector, the operator norm of the matrix~$\MM$ is necessarily bounded by~$1$. We have the following stronger property, which is proved in~\cite{HJ17}.

\begin{lemma}[Lemma~4.11 in~\cite{HJ17}]\label{lem:spr}
	Under assumptions (iii) and (iv), \[\spr \MM < 1.\]
	In particular $\|\MM^t\|$ converges exponentially fast to $0$ as $t \to \infty$ along~$\nn$.
\end{lemma}

We also introduce the shorthand~$\h_\env^+$ for the closure of the linear span of vectors of the form~$\delta_\ell \otimes w$ for $\ell \geq 0$ and arbitrary $w \in \cc^m$. It corresponds to one side of the bi-infinite environment. Combining Lemmas~\ref{lem:one-tau} and~\ref{lem:spr}, we have the following proposition on the long-time evolution of pairs of creation and annihilation operators.

\begin{proposition}\label{prop:evol-quad}
	For all~$t \in \nn$, under assumptions (ii)--(iv), we have:
	\begin{enumerate}
		\item for all $\psi_1,\psi_2 \in \h_\sys$,
		\begin{equation*}
		\begin{split}
			\tau^t(a^{\sharp_1}(\psi_1) a^{\sharp_2}(\psi_2))
				&= \sin^2\alpha\,\sum_{s',t'=1}^t b^{\sharp_1}((S \otimes U)^{s'-t}\iota W^* {\MMs}^{s'-1} \psi_1)
				\\ &\qquad\qquad\qquad b^{\sharp_2}((S \otimes U)^{t'-t}\iota W^* {\MMs}^{t'-1} \psi_2) + O(t\|\MM^t\|);
		\end{split}
		\end{equation*}
		\item for all $\varphi_1 \in \h_\env^+$ and $\psi_2 \in \h_\sys$,
		\begin{equation*}
		\begin{split}
			\tau^t(b^{\sharp_1}(\varphi_1) a^{\sharp_2}(\psi_2))
				& = -\ii\sin\alpha\, \sum_{t'=1}^t b^{\sharp_1}((S \otimes U)^{-t} \varphi_1) \\ &\qquad\qquad\qquad  b^{\sharp_2}((S \otimes U)^{t'-t}\iota W^* {\MMs}^{t'-1} \psi_2) + O(t\|\MM^t\|);
		\end{split}
		\end{equation*}
		\item for all $\varphi_1,\varphi_2 \in \h_\env^+$,
		\begin{equation*}
			\tau^t(b^{\sharp_1}(\varphi_1) b^{\sharp_2}(\varphi_2))
				= b^{\sharp_1}((S \otimes U)^{-t} \varphi_1) b^{\sharp_2}((S \otimes U)^{-t} \varphi_2).
		\end{equation*}
	\end{enumerate}
	The notation $O(t\|\MM^t\|)$ stands for error terms which are bounded in norm by $t\|\MM^t\|$ times a numerical constant that is independent of the vectors under consideration, as long as they are normalized.
\end{proposition}

\begin{proof}
	In the second and third part of the proposition, we are considering $\varphi_i \in \h_\env^+$ because those are the ones that appear in the time evolution of pairs of creation and annihilation operators in the sample (see Lemma~\ref{lem:one-tau}). They have the property that $\iota^* S^*U^*\varphi_i = 0$, which makes the computations more tractable.
  We prove the claims in a different order.
	\begin{enumerate}
		\item[3.] The formula follows directly from applying the formula in Proposition~\ref{lem:one-tau} and the identity $\iota^* S^*U^*\varphi_i = 0$.
		\item[2.] Using $\iota^* S^*U^*\varphi_1 = 0$, Proposition~\ref{lem:one-tau} yields
			\begin{align*}
				\tau(b^{\sharp_1}(\varphi_1)a^{\sharp_2}(\psi_2)) &= b^{\sharp_1}(S^*U^*\varphi_1)a^{\sharp_2}(\MMs\psi_2)
        \\
        &\qquad{}- \ii\sin\alpha\,b^{\sharp_1}(S^*U^*\varphi_1)b^{\sharp_2}(\iota W^*\psi_2),
			\end{align*}
			and again $S^*U^*\varphi_1 \in \h_\env^+$. Hence, for any $t \in \nn$,
			\begin{align*}
				\tau^t(b^{\sharp_1}(\varphi_1)a^{\sharp_2}(\psi_2)) &= b^{\sharp_1}((S^*U^*)^t\varphi_1)a^{\sharp_2}((\MMs)^t\psi_2)
        \\ &\qquad{} - \sum_{t'=1}^t\ii\sin\alpha\,b^{\sharp_1}((S^*U^*)^{t}\varphi_1)
        b^{\sharp_2}((S^*U^*)^{t-t'}\iota W^* {\MMs}^{t'-1}\psi_2).
			\end{align*}
			Because $\|a^\sharp(\psi)\| = \|\psi\|$ for any~$\psi \in \h_\sys$ and similarly for $b^{\sharp}$, we have the bound
			\[
				\|b^{\sharp_1}((S^*U^*)^t\varphi_1)a^{\sharp_2}((\MMs)^t\psi_2)\| \leq \|(S^*U^*)^t\varphi_1\|\|(\MMs)^t\psi_2\| \leq \|\MM^t\|.
			\]

		\item[1.] Again by Proposition~\ref{lem:one-tau},
		\begin{align*}
			\tau(a^{\sharp_1}(\psi_1)a^{\sharp_2}(\psi_2)) &= a^{\sharp_1}(\MMs\psi_1)a^{\sharp_2}(\MMs\psi_2)
        + \ii\sin\alpha\,a^{\sharp_1}(\MMs\psi_1) b^{\sharp_2}(\iota W^*\psi_2)
				\\
				&\qquad
				- \ii\sin\alpha\,b^{\sharp_1}(\iota W^*\psi_1)a^{\sharp_2}(\MMs\psi_2)
        + \sin^2\alpha\,b^{\sharp_1}(\iota W^*\psi_1)b^{\sharp_2}(\iota W^*\psi_2),
		\end{align*}
		with $\iota W^*\psi_1, \iota W^*\psi_2 \in \h_\env^+$. Hence, for any $t \in \nn$,
		\begin{align*}
			&\tau(a^{\sharp_1}(\psi_1)a^{\sharp_2}(\psi_2)) \\
        &\qquad = a^{\sharp_1}({\MMs}^t\psi_1)a^{\sharp_2}({\MMs}^t\psi_2) \\
				&\qquad\qquad + \sum_{t'=1}^t\ii\sin\alpha\, a^{\sharp_1}({\MMs}^t\psi_1) b^{\sharp_2}((S^*U^*)^{t-t'}\iota W^* {\MMs}^{t'-1}\psi_2)\\
				&\qquad\qquad - \sum_{s'=1}^t\ii\sin\alpha\, b^{\sharp_1}((S^*U^*)^{t-s'}\iota W^*{\MMs}^{s'-1}\psi_1)a^{\sharp_2}({\MMs}^t\psi_2) \\
				&\qquad\qquad + \sum_{s'=1}^t\sum_{t'=1}^t\sin^2\alpha\, b^{\sharp_1}((S^*U^*)^{t-s'}\iota W^* (\MMs)^{s'-1}\psi_1)
				\\&\qquad\qquad\qquad\qquad b^{\sharp_2}((S^*U^*)^{t-t'}\iota W^*(\MMs)^{t'-1}\psi_2).
		\end{align*}
		The norm estimates are obtained with similar bounds as in the previous part. \qedhere
	\end{enumerate}
\end{proof}

The main result of this section is the following theorem. It states that, under our ongoing assumptions, the limiting state of the sample~$\sys$ is \textsc{giqf} and completely described by its own dynamics~$W$ and the same scalar functions~$(\fonci)_{i=1}^{m}$, up to an interaction-dependent deformation which vanishes in the limit~$\alpha \to 0$.

\begin{theorem}\label{thm:q-f}
	Under hypotheses (i)--(iv) stated at the beginning of the section, the limit
	$$
		\rho_\infty(A) := \lim_{t\to\infty} (\omega_\Sigma \otimes \rho)(\tau^t(\one \otimes A))
	$$
	exists for all $A \in \CAR(\h_\sys)$ and~$\rho_\infty$ defines a \textsc{giqf} state on~$\sys$ with symbol
	\begin{align}
		\Delta &:=
      \sum_{i=1}^{m} \|\pi_i v\|^2 2 \Re\fonci(\MMs),\label{eq:alt-delta}
	\end{align}
	where $\MM = W(\one + (1-\cos\alpha)P)$. Moreover, the convergence $\tr_{\CAR(\h_\env)}((\omega_\Sigma \otimes \rho) \circ \tau^t) \to \rho_\infty$ happens exponentially fast in norm.
\end{theorem}

Before we proceed with the proof, let us make a few comments on the result.
	The fact that the asymptotic density~$\Delta$ in the sample~$\sys$ can be rewritten in terms of the functions $\fonci$ defined by~\eqref{eq:sigma-plus} (also see Remark~\ref{rem:int-sigma}) is to be compared with the formula
	\begin{equation*}\label{eq:res-res}
		\Sigma 
		= \sum_{i=1}^{m} 2 \Re\fonci (S^* \otimes U^*)(\one\otimes\pi_i)
	\end{equation*}
	for the symbol for the initial state~$\omega_\Sigma$ of the environment~$\env$.

	Note that~$\Delta$ depends on the coupling constant~$\alpha$ through the matrix~$\MM$. In the small coupling limit $\alpha \to 0$, we have $\MMs \to W^*$ and thus
	\[
		\Delta \to \sum_{i=1}^{m} \|\pi_i v\|^2 2 \Re\fonci(W^*).
	\]
	In this regime, the asymptotic of the sample is completely determined by its free dynamics, the functions defining the initial state of the environment, and the ratios between the coupling with the different parts of the environment. This state is of course invariant under the free dynamics in the sample~$\sys$.

	Typically, for  $\psi_1, \psi_2 \in \h_\sys$,
	\begin{equation}
		(\omega_\Sigma \otimes \rho_\infty)(\tau(a^*(\psi_1)a(\psi_2))) \neq  \rho_\infty (a^*(\psi_1)a(\psi_2)).
	\end{equation}
	This is not surprising since it can be seen from the formula~\eqref{eq:one-aa-tau} that the reduced dynamics corresponding to one step of time evolution with a fresh environment only sees the first coefficient in the expansion for the state of the environment.
	However, equality holds for all $\psi_1, \psi_2 \in \h_\sys$ when the functions~$\fonci$ are constant\,---\,that is when there are no correlations in the environment~$\env$. Also, in any case, both sides converge to the same quantity in the limit~$\alpha \to 0$.

\begin{proof}[Proof of Theorem~\ref{thm:q-f}]
	It is sufficient to prove the result for obervables~$A$ which are monomials in creation and annihilation operators in the sample~$\sys$. The convergence in norm will then be immediate from the convergence of the matrix elements because we are working on a CAR for a finite number of degrees of freedom.

	We proceed in four steps. First, we show that the fact that the initial state is a tensor product of two even states implies that $A \mapsto (\omega_\Sigma \otimes \rho)(\tau^t(\one \otimes A))$ is an even state on~$\CAR(\h_\sys)$ for all $t\geq 0$. Then, we need only consider the asymptotic evolution of monomials of even degree in~$a$ and~$a^*$. We simplify the formula at the last step.

	Throughout the proof, whenever a product~$\prod_{i} a^{\sharp_i}(e_{k_i})$ appears in a formula, it is ordered with the term for~$i+1$ to the right of the term for~$i$.
	\begin{description}
		\item[Step 1: the asymptotic state is even.]  Since the dynamics~$(\tau^t)_{t \in \nn}$ preserves the total number of particles, $\tau$ maps monomials of odd degree in~$a$ and~$a^*$ to a linear combination of monomials of odd degree in~$a$,~$a^*$,~$b$ and~$b^*$. Then, each monomial is either of odd degree in~$a$ and~$a^*$ or of odd degree in~$b$ and~$b^*$. Since both~$\rho$ and~$\omega_\Sigma$ are even states, this implies
		\begin{equation*}
				(\omega_\Sigma \otimes \rho)\Big(\tau^t\Big(\prod_{i=1}^{2N+1} a^{\sharp_i}(\psi_i)\Big)\Big) = 0.
		\end{equation*}
		and thus
		\begin{equation}\label{eq:odd-state}
				\rho_\infty\Big(\prod_{i=1}^{2N+1} a^{\sharp_i}(\psi_i)\Big) = \lim_{t\to\infty}(\omega_\Sigma \otimes \rho)\Big(\tau^t\Big(\prod_{i=1}^{2N+1} a^{\sharp_i}(\psi_i)\Big)\Big) = 0.
		\end{equation}
		for any~$m \in \nn$ and any choices of $\psi_1, \dotsc, \psi_{2N+1} \in \h_\sys$.

		\item[Step 2: the asymptotic time evolution of monomials of even degree.]
		Using the appropriate formula from Proposition~\ref{prop:evol-quad} and the fact that $\tau^t$ is a morphism, we have
		\begin{equation}\label{eq:evol-pairs}
				\tau^t\Big(\prod_{i=1}^{2N} a^{\sharp_i}(\psi_i)\Big) = \prod_{i=1}^{2N} \Big( \sum_{t_i=1}^t\sin\alpha\, b^{\sharp_i}((S \otimes U)^{t_i-t}\iota W^* {\MMs}^{t_i-1}  \psi_i)  + O(t\|\MM^t\|)\Big).
		\end{equation}

		\item[Step 3: the asymptotic state is a gauge-invariant quasi-free state.] From the previous step, we have
		\begin{align*}
				\rho_\infty\Big(\prod_{i=1}^{2N} a^{\sharp_i}(\psi_i)\Big)&= \lim_{t\to\infty}(\omega_\Sigma \otimes \rho)\Big(\tau^t\Big(\prod_{i=1}^{2N} a^{\sharp_i}(\psi_i)\Big)\Big) \\
				&= \lim_{t\to\infty}\omega_\Sigma\Big(\prod_{i=1}^{2N} \sum_{t_i=1}^t\sin\alpha\, b^{\sharp_i}((S \otimes U)^{t_i-t}\iota W^* {\MMs}^{t_i-1}  \psi_i)\Big).
		\end{align*}
		Using the definition of~$\omega_\Sigma$ as a gauge-invariant quasi-free state with density~$\Sigma$,
		\begin{align*}
				&\rho_\infty(a^{*}(\psi_1) \dotsb a^{*}(\psi_N)a(\psi'_N) \dotsb a(\psi'_1)) \\
				&\qquad=
					 \lim_{t \to\infty}\det\Big[ \sum_{s',t'=1}^{t} \sin^2\alpha \\
					 &\qquad\qquad\qquad \Big\langle (S \otimes U)^{t'-t}\iota W^* {\MMs}^{t'-1}  \psi'_\nu, \Sigma (S \otimes U)^{s'-t}\iota W^* {\MMs}^{s'-1}  \psi_\mu \Big\rangle \Big]_{\mu,\nu=1}^{N}.
		\end{align*}
		Using $[\Sigma,S^* \otimes U^*]=0$, and omitting the details of the indexation of the matrix considered for the determinant,
		\begin{align*}
				&\rho_\infty(a^{*}(\psi_1) \dotsb a^{*}(\psi_N)a(\psi'_N) \dotsb a(\psi'_1)) \\
				&\qquad=
					 \lim_{t \to\infty}\!\det\!\Big[ \sum_{s',t'=1}^{t} \sin^2\alpha \,
					 \Big\langle \psi'_\nu, {\MM}^{t'-1}W \iota^* \Sigma (S \otimes U)^{s'-t'}\iota W^* {\MMs}^{s'-1}  \psi_\mu \Big\rangle \Big]_{\mu,\nu} \\
			 	&\qquad=
				  	\det\Big[ \Big\langle \psi'_\nu, \sum_{s',t'=1}^{\infty} \sin^2\alpha \\
				  	&\qquad\qquad  \braket{\delta_0 \otimes v, \Sigma (S^{s'-t'} \delta_0 \otimes U^{s'-t'}v) }{\MM}^{t'-1}W \iota^*\iota W^* {\MMs}^{s'-1}  \psi_\mu \Big\rangle \Big]_{\mu,\nu}.
		\end{align*}
		The fact that
		$$
			\rho_\infty(a^{*}(\psi_1) \dotsb a^{*}(\psi_N)a(\psi'_{N'}) \dotsb a(\psi'_1)) = 0
		$$
		whenever $N \neq N'$ follows from~\eqref{eq:odd-state} of Step~1 if $N+N'$ is odd, and from~\eqref{eq:evol-pairs} of Step~2 and the fact that~$\omega_\Sigma$ is a gauge-invariant quasi-free state if~$N+N'$ is even. We conclude that $\rho_\infty: A \mapsto \lim_{t\to\infty}(\omega_\Sigma \otimes \rho)(\tau^t(\one \otimes A))$ is a gauge-invariant quasi-free state with density
		\begin{equation*}
			\sum_{s',t'=1}^{\infty} \braket{\delta_0 \otimes v, \Sigma (S^{s'-t'} \delta_0 \otimes U^{s'-t'}v) }{\MM}^{t'-1}\sin^2\alpha\,W P W^* {\MMs}^{s'-1}.
		\end{equation*}

		\item[Step 4: the alternate formula.]  By definition,~$\MM=W(\one + (\cos\alpha - 1)P)$. Using the fact that~$W$ is unitary and basic trigonometry, we find the identity
		\[
			\sin^2\alpha\, W P W^* = \one - \MM \MMs.
		\]
		The formula~\eqref{eq:alt-delta} for~$\Delta$ given in the statement of the proposition then follows from a telescoping and reindexing of the summation. \qedhere
	\end{description}
\end{proof}

This knowledge of the asymptotic state in the sample~$\sys$ allows us to investigate the number of particles there. We give more details on the particle number at each node~$\nu$ of the graph in a more concrete example in Section~\ref{sec:ex}.

\begin{corollary}
	Under assumptions (i)--(iv), in the limit $t \to \infty$, the number of particles in the sample is distributed as a sum of~$d$ independent Bernoulli random variables with parameters~$\lambda_0, \dotsc, \lambda_{d-1} \in (0,1)$ that are the eigenvalues of the self-adjoint matrix~$\Delta$.
\end{corollary}

\begin{proof}
	By a standard continuity argument~(see {e.g.}~\cite[\S{IV.A}]{DFP08}), we need only consider the case~$0 < \Delta < \one$.

	By standard results on \textsc{giqf} states (see {e.g.}~\cite[\S{4.7.3}]{JOPP} or~\cite[\S{IV.A}]{DFP08}), the quasi-free state is associated to the density matrix
	$$
		\rho_\infty = \det(\one - \Delta) \bigoplus_{p=0}^d \Big(\frac{\Delta}{\one-\Delta}\Big)^{\wedge p}.
	$$
	Therefore, the probability of observing $p$ particles in the sample is given by
	\begin{align*}
		\mathbf{P}(p)
			&= \det(\one - \Delta) \tr \Big(\Big(\frac{\Delta}{\one-\Delta}\Big)^{\wedge p}\Big).
	\end{align*}
	Diagonalizing~$\Delta$, labeling its eigenvalues~$\lambda_0, \dots, \lambda_{d-1}$ and using cyclicity of the trace yields
	\begin{align*}
		\mathbf{P}(p)
			&= (1-\lambda_0) \dotsb (1-\lambda_{d-1}) \sum_{k_1 < \dotsb < k_p} \frac{\lambda_{k_1} \dotsb \lambda_{k_d}}{(1-\lambda_{k_1}) \dotsb (1-\lambda_{k_d})} \\
			&= \sum_{k_1 < \dotsb < k_p} {\lambda_{k_1} \dotsb \lambda_{k_d}} \prod_{j\notin\{k_n\}_{m=1}^p} (1 - \lambda_j).
	\end{align*}
	This probability mass function is precisely that of a sum of~$d$ independent Bernoulli random variables~$(X_i)_{i=1}^d$ where $X_i \sim \operatorname{Ber}(\lambda_i)$, also known as Poisson binomial distribution of parameter~$(\lambda_0, \dotsc, \lambda_{d-1})$.
\end{proof}

\subsection{Flux observables}

Recall that~$U$ is a unitary on~$\cc^m$ and~$\{\pi_i\}_{i=1}^m$ is its set of spectral rank-one projectors associated to distinct eigenvalues~$\{\Exp{\ii\gamma_i}\}_{i=1}^m$. The commutation relations $$[\Sigma, \one\otimes \pi_i] = 0$$ for each~$i$
suggest a decomposition of the Hilbert space~$\h_\env$ into~$m$ infinte-dimensional subspaces, in which we can formally count the number of particles. Lemma~\ref{lem:one-tau} allows us to formally compute the difference
\[
  \sum_{\ell \in \zz} \tau(b^*(\delta_\ell \otimes x_i)b(\delta_\ell \otimes x_i)) - b^*(\delta_\ell \otimes x_i)b(\delta_\ell \otimes x_i)
\]
in the number of particles in those subspaces between two time steps.
The result of this computation yields a bona fide bounded self-adjoint operator and we define
\begin{align*}
	\Phi_i &:=  (\cos\alpha-1)^2\|\pi_i v\|^2 b^*(\delta_0 \otimes v)b(\delta_0 \otimes v) \\
		&\qquad
		+ (\cos\alpha - 1)b^*(\delta_0 \otimes \pi_i v)b(\delta_0 \otimes v)  + (\cos\alpha - 1)b^*(\delta_0 \otimes v)b(\delta_0 \otimes \pi_i v) \\
		&\qquad
		+\ii\sin\alpha\, (\cos\alpha-1)\|\pi_i v\|^2\big( b(\delta_0 \otimes v)a^*(\psi_*) - b^*(\delta_0 \otimes v)a(\psi_*) \big) \\
		&\qquad
		+ \ii\sin\alpha\,\big(b(\delta_0 \otimes \pi_i v) a^*(\psi_*)  - b^*(\delta_0 \otimes \pi_i v) a(\psi_*)\big) \\
		&\qquad
		+ \sin^2\alpha\,\|\pi_i v\|^2a^*(\psi_*)a(\psi_*),
\end{align*}
the flux observable into the $i$th subreservoir, $\env_i$, accordingly. Summing over~$i=1,2,\dotsc,m$, we get the observable
\begin{align*}
	\Phi_\env &:= (\cos^2\alpha - 1) b^*(\delta_0 \otimes v)b(\delta_0 \otimes v) + \sin^2\alpha\,a^*(\psi_*)a(\psi_*) \\
		&\qquad
		+\ii\sin\alpha\,\cos\alpha\,\big( b(\delta_0 \otimes v)a^*(\psi_*) - b^*(\delta_0 \otimes v)a(\psi_*) \big)
		,
\end{align*}
for the flux into the whole environment. This is in agreement\footnote{{up to a change of sign of the coupling constant~$\alpha$}} with formula~(28) in~\cite{HJ17} and with a similar computation that can be done for the flux out of the sample. Of course, $\lim_{t\to\infty}(\omega_\Sigma \otimes \rho) \tau^t(\Phi_\env) = 0$.

In order to compute the limiting expectation of the flux operators~$\Phi_1, \dotsc, \Phi_m$, we need the long-time evolution of the quadratic monomials appearing in the defining formula. The following result follows directly from Proposition~\ref{prop:evol-quad} and the definition of~$\omega_\Sigma$.

\begin{corollary}\label{cor:limit-pairs-state}
		Under assumptions (i)--(iv), we have
		\begin{enumerate}
			\item for all $\psi_1,\psi_2 \in \h_\sys$,
			\begin{equation*}
			\begin{split}
				&\lim_{t\to\infty}(\omega_\Sigma \otimes \rho)\tau^t(a^*(\psi_1) a(\psi_2)) \\
					&\qquad = \sum_{s',t'=1}^\infty \sin^2\alpha\,\braket{S^{t'}\delta_0 \otimes U^{t'}v, \Sigma (S^{s'}\delta_0 \otimes U^{s'}v)}
          \braket{\psi_2, {\MM}^{t'-1}W P W^* {\MMs}^{s'-1} \psi_1} ;
			\end{split}
			\end{equation*}
			\item for all $\varphi_1 \in \h_\env^+$ and $\psi_2 \in \h_\sys$,
			\begin{equation*}
			\begin{split}
				&\lim_{t\to\infty}(\omega_\Sigma \otimes \rho)\tau^t(b^*(\varphi_1) a(\psi_2))
          = -\ii\sin\alpha\, \sum_{t'=1}^\infty \braket{S^{t'}\delta_0 \otimes U^{t'}, \Sigma \varphi_1} \overline{\braket{\psi_*, W^* {\MMs}^{t'-1} \psi_2}};
			\end{split}
			\end{equation*}
			\item for all $\varphi_1,\varphi_2 \in \h_\env^+$,
			\begin{equation*}
				\lim_{t\to\infty} (\omega_\Sigma \otimes \rho)\tau^t(b^*(\varphi_1) b(\varphi_2))
					= \braket{\varphi_2, \Sigma \varphi_1}.
			\end{equation*}
		\end{enumerate}
\end{corollary}

We may then compute the asymptotics of the flux $\Phi_i$ into the $i$th part of the environment. If the coupling constant $\alpha$ is small enough, we may then determine the sign of this flux by comparing $\fonci(1)$ with the different $\foncj(1)$, $j \neq i$, weighted by the appropriate scalar products; see the remark below.

\begin{proposition}\label{prop:flux}
	Under assumptions (i)--(iv), we have
	\begin{equation*}
		\begin{split}
			&\lim_{t\to\infty}(\omega_\Sigma \otimes \rho)\tau^t(\Phi_i) \\
				&=  (2-2\cos\alpha) \big(\|\pi_i v\|^2\braket{\delta_0 \otimes v, \Sigma(\delta_0 \otimes v)} - \braket{\delta_0 \otimes \pi_i v, \Sigma(\delta_0 \otimes \pi_i v)}\big) \\
				&\qquad
				+2\Re\sin^2\alpha\,\sum_{t'=1}^\infty \braket{\psi_*,  {\MM}^{t'-1} W \psi_*} \\
					&\qquad\qquad \big(\|\pi_i v\|^2\braket{S^{t'}\delta_0\otimes U^{t'}v ,\Sigma(\delta_0 \otimes v)} - \braket{S^{t'}\delta_0 \otimes U^{t'}\pi_i v, \Sigma(\delta_0 \otimes \pi_i v)}\big) .
	\end{split}
	\end{equation*}
\end{proposition}

\begin{remark}
	For $\alpha \ll 1$, we have $2 - 2 \cos \alpha = \alpha^2 + O(\alpha^3)$ and hence
	\begin{equation*}
		\begin{split}
			\lim_{\alpha\to0}\lim_{t\to\infty}\frac{(\omega_\Sigma \otimes \rho)\tau^t(\Phi_i)}{\alpha^2}
				&=
					\|\pi_i v\|^2(1-\|\pi_i v\|^2)
          2\Re \Big(\Big(\sum_{j \neq i} \tfrac{\|\pi_j v\|^2}{1-\|\pi_i v\|^2} \foncj(1) \Big) -\fonci(1)\Big),
	\end{split}
	\end{equation*}
	assuming~$v \neq \pi_i v$ for each~$i$. Therefore, for small enough coupling, the sign of the flux into the $i$th subreservoir is given by that of
	\[
		\Big(\sum_{j \neq i} \tfrac{\|\pi_j v\|^2}{1-\|\pi_i v\|^2} \Re\foncj(1) \Big) - \Re\fonci(1),
	\]
	which is simply the sign of
	$
		\Re\foncj(1) - \Re\fonci(1)
	$
	in the case of two subreservoirs. This gives interpretation of the number $2\Re\fonci(1)$. Recall that~$2\fonci(0)$ is the average particle density in the $i$th sub reservoir.
\end{remark}

\begin{proof}[Proof of Proposition~\ref{prop:flux}]
	From the definition of~$\Phi_i$
  {
	\begin{align*}
		\lim_{t\to\infty}(\omega_\Sigma \otimes \rho)\tau^t(\Phi_i)
			&=  (\cos\alpha-1)^2\|\pi_i v\|^2 \lim_{t\to\infty}(\omega_\Sigma \otimes \rho)\tau^t(b^*(\delta_0 \otimes v)b(\delta_0 \otimes v)) \\
			&\qquad
			+ (\cos\alpha - 1)\lim_{t\to\infty}(\omega_\Sigma \otimes \rho)\tau^t(b^*(\delta_0 \otimes \pi_i v)b(\delta_0 \otimes v))  \\
			&\qquad
			+ (\cos\alpha - 1)\lim_{t\to\infty}(\omega_\Sigma \otimes \rho)\tau^t(b^*(\delta_0 \otimes v)b(\delta_0 \otimes \pi_i v)) \\
			&\qquad
			+\ii\sin\alpha\, (\cos\alpha-1)\|\pi_i v\|^2
        \lim_{t\to\infty}(\omega_\Sigma \otimes \rho)\tau^t(b(\delta_0 \otimes v)a^*(\psi_*)) \\
			&\qquad
			-\ii\sin\alpha\, (\cos\alpha-1)\|\pi_i v\|^2
        \lim_{t\to\infty}(\omega_\Sigma \otimes \rho)\tau^t(b^*(\delta_0 \otimes v)a(\psi_*))  \\
			&\qquad
			+ \ii\sin\alpha \lim_{t\to\infty}(\omega_\Sigma \otimes \rho)\tau^t(b(\delta_0 \otimes \pi_i v) a^*(\psi_*))   \\
			&\qquad
			- \ii\sin\alpha\lim_{t\to\infty}(\omega_\Sigma \otimes \rho)\tau^t(b^*(\delta_0 \otimes \pi_i v) a(\psi_*)) \\
			&\qquad
			+ \sin^2\alpha\,\|\pi_i v\|^2\lim_{t\to\infty}(\omega_\Sigma \otimes \rho)\tau^t(a^*(\psi_*)a(\psi_*)).
	\end{align*}
  }
  Using Corollary~\ref{cor:limit-pairs-state} for each term gives
  {
	\begin{align*}
		&\lim_{t\to\infty}(\omega_\Sigma \otimes \rho)\tau^t(\Phi_i) \\
			&=  (\cos\alpha-1)^2\|\pi_i v\|^2 \braket{\delta_0 \otimes v, \Sigma(\delta_0 \otimes v)}
			+ 2  (\cos\alpha - 1)\Re\braket{\delta_0 \otimes v, \Sigma(\delta_0 \otimes \pi_i v)} \\
			&\qquad
			-2\sin^2\alpha\, (\cos\alpha-1)\|\pi_i v\|^2
        \Re\sum_{t'=1}^\infty \overline{\braket{\psi_*, W^* {\MMs}^{t'-1}\psi_*}}\braket{S^{t'}\delta_0\otimes U^{t'}v ,\Sigma(\delta_0 \otimes v)}  \\
			&\qquad
			-2\sin^2\alpha\,
      \Re\sum_{t'=1}^\infty \overline{\braket{\psi_*, W^*{\MMs}^{t'-1}\psi_*}} \braket{S^{t'}\delta_0 \otimes U^{t'}v, \Sigma (\delta_0 \otimes \pi_i v)} \\
			&\qquad
			+ \sin^4\alpha\,\|\pi_i v\|^2
      \\
      &\qquad\qquad
      \sum_{s',t'=1}^\infty \braket{\delta_0 \otimes v, \Sigma (S^{s'-t'} \delta_0 \otimes U^{s'-t'}v) }
				\braket{\psi_*,{\MM}^{t'-1}W \iota^*\iota W^* {\MMs}^{s'-1}\psi_*}.
	\end{align*}
  }
	Using the commutation relation $[\Sigma, \one\otimes \pi_i ] = 0$, we have
  {
	\begin{align*}
		&\lim_{t\to\infty}(\omega_\Sigma \otimes \rho)\tau^t(\Phi_i) \\
			&=  (\cos\alpha-1)^2\|\pi_i v\|^2 \braket{\delta_0 \otimes v, \Sigma(\delta_0 \otimes v)}
			+ 2 (\cos\alpha - 1) \braket{\delta_0 \otimes \pi_i v, \Sigma(\delta_0 \otimes \pi_i v)} \\
			&\qquad
			-2\sin^2\alpha\, (\cos\alpha-1)\|\pi_i v\|^2\Re\sum_{t'=1}^\infty \braket{\psi_*,  {\MM}^{t'-1} W \psi_*} \braket{S^{t'}\delta_0\otimes U^{t'}v ,\Sigma(\delta_0 \otimes v)}  \\
			&\qquad
			-2\sin^2\alpha\, \Re\sum_{t'=1}^\infty \braket{\psi_*, {\MM}^{t'-1} W \psi_*} \braket{S^{t'}\delta_0 \otimes U^{t'} \pi_i v, \Sigma(\delta_0 \otimes \pi_i v)} \\
			&\qquad
			+ \sin^4\alpha\,\|\pi_i v\|^2
      \\
      &\qquad\qquad
        \sum_{s',t'=1}^\infty \braket{\delta_0 \otimes v, \Sigma (S^{s'-t'} \delta_0 \otimes U^{s'-t'}v) }
				\braket{\psi_*,{\MM}^{t'-1}W P W^* {\MMs}^{s'-1}\psi_*}.
	\end{align*}
  }
	The identity
	$
		\sin^2\alpha\, W P W^* = \one - \MM\MMs
	$
	then implies
  {
	\begin{align*}
		&\lim_{t\to\infty}(\omega_\Sigma \otimes \rho)\tau^t(\Phi_i) \\
			&\qquad=  (\cos\alpha-1)^2\|\pi_i v\|^2 \braket{\delta_0 \otimes v, \Sigma(\delta_0 \otimes v)}
			+ 2 (\cos\alpha - 1) \braket{\delta_0 \otimes \pi_i v, \Sigma(\delta_0 \otimes \pi_i v)} \\
			&\qquad\qquad
			-2\sin^2\alpha\, (\cos\alpha-1)\|\pi_i v\|^2
        \Re\sum_{t'=1}^\infty \braket{\psi_*,  {\MM}^{t'-1} W \psi_*} \braket{S^{t'}\delta_0\otimes U^{t'}v ,\Sigma(\delta_0 \otimes v)}  \\
			&\qquad\qquad
			-2\sin^2\alpha\, \Re\sum_{t'=1}^\infty \braket{\psi_*, {\MM}^{t'-1} W \psi_*} \braket{S^{t'}\delta_0 \otimes U^{t'}\pi_i v, \Sigma(\delta_0 \otimes \pi_i v)} \\
			&\qquad\qquad
			+\sin^2\alpha\,\|\pi_i v\|^2 \braket{\delta_0 \otimes v, \Sigma \delta_0 \otimes v) }\\
			&\qquad\qquad
			+ 2\sin^2\alpha\,\|\pi_i v\|^2 \Re\sum_{t'=1}^\infty \braket{S^{t'}\delta_0 \otimes U^{t'}v, \Sigma (\delta_0 \otimes v)}
				\braket{\psi_*,{\MM}^{t'}\psi_*}.
	\end{align*}
  }
	Using the identities
	$
		\sin^2\alpha + (\cos\alpha-1)^2 = 2 - 2\cos\alpha
	$
	and
	$
		\cos\alpha\,W \psi_* = \MM \psi_*,
	$
	we conclude
	\begin{align*}
		&\lim_{t\to\infty}(\omega_\Sigma \otimes \rho)\tau^t(\Phi_i) \\
			&\qquad=  -2(\cos\alpha-1)\|\pi_i v\|^2 \braket{\delta_0 \otimes v, \Sigma(\delta_0 \otimes v)}
			\\
			&\qquad\qquad
			+ 2 (\cos\alpha - 1) \braket{\delta_0 \otimes \pi_i v, \Sigma(\delta_0 \otimes \pi_i v)} \\
			&\qquad\qquad
			+2\sin^2\alpha\,\|\pi_i v\|^2\Re\sum_{t'=1}^\infty \braket{\psi_*,  {\MM}^{t'-1} W \psi_*} \braket{S^{t'}\delta_0\otimes U^{t'}v ,\Sigma(\delta_0 \otimes v)}  \\
			&\qquad\qquad
			-2\sin^2\alpha\, \Re\sum_{t'=1}^\infty \braket{\psi_*, {\MM}^{t'-1} W \psi_*} \braket{S^{t'}\delta_0 \otimes U^{t'}\pi_i v, \Sigma(\delta_0 \otimes \pi_i v)}.
	\qedhere
	\end{align*}
\end{proof}

\section{Examples} \label{sec:ex}

To give more detailed information on the profile of the particle density on the graph and to give further interpretation of other related quantities, we restrict our attention to simple models in which $m=1$.

\subsection{Quantum walks on a ring}

Consider the first example, the case of spin-$\tfrac 12$ quantum walkers on a cycle of~$n$ vertices. We use the shorthands~$e_{\nu,\pm}$ for $\delta_\nu \otimes e_{\pm 1}$ and~$n_{\nu}$ for~$a^*(e_{\nu,+}) a(e_{\nu,+}) + a^*(e_{\nu,-}) a(e_{\nu,-})$.
Under the hypotheses of our results, the particle density at the vertex~$\nu$,
\[
  p_t(\nu) := (\omega_\Sigma \otimes \rho) \tau^t( n_\nu),
\]
converges to
\[
  \rho_\infty(n_\nu)= \braket{e_{\nu,+}, \Delta e_{\nu,+}} + \braket{e_{\nu,-}, \Delta e_{\nu,-}}
\]
as $t\to\infty$. We refer to the function of the node~$\nu$ defined by this limit as the profile~$p$.

Note that the block structure of~$\MM$ implies $\braket{e_{\nu,\pm}, \MM^{2j+1} e_{\nu,\pm}} = 0$ for each $j \in \nn$.
Therefore, the asymptotic profile~$p$ of the particle density in the sample is independent of the odd coefficients in the series~\eqref{eq:sigma-plus} describing the symbol~$\Sigma$ (the initial state in the reservoir) via~\eqref{eq:sigma-series}.

We can also consider the position correlations
\begin{align*}
	C_t(\nu,\upsilon) &= (\omega_\Sigma \otimes \rho)(\tau^t(n_{\nu}n_{\upsilon})) - (\omega_\Sigma \otimes \rho)(\tau^t(n_{\nu}))(\omega_\Sigma \otimes \rho)(\tau^t(n_{\upsilon})).
\end{align*}
In the large time limit, a standard computation shows
\begin{align*}
	\lim_{t\to\infty}C_t(\nu,\upsilon) &= \rho_\infty(n_\nu n_\upsilon) - \rho_\infty(n_\nu) \rho_\infty(n_\upsilon)\\
  &= -\sum_{\tau_\nu,\tau_\upsilon \in \{+,-\}}|\braket{e_{\nu,\tau_\nu}, \Delta e_{\upsilon,\tau_\upsilon}}|^2
\end{align*}
for $\nu \neq u$. Note the definite sign.

\begin{example}
	Consider the case where each coin unitary is a rotation matrix of angle $\theta_\nu \notin \tfrac{\pi}{2}\zz$, {i.e.} $C_\nu =(\begin{smallmatrix}\cos\theta_\nu & -\sin\theta_\nu \\ \sin\theta_\nu & \cos\theta_\nu \end{smallmatrix})$. If the series~\eqref{eq:sigma-plus} terminates after the quadratic term ($\ell=2$), we can compute explicitly
	$$
		{p(\nu)} =
		\begin{cases}
      2F(0) - \Re F^{(2)}(0) (\cos\alpha \sin\theta_{n} \sin\theta_1 + \sin\theta_1 \sin\theta_2) &  1 = \nu  \\
			2F(0) - \Re F^{(2)}(0) (\sin\theta_{\nu-1}\sin\theta_\nu + \sin\theta_\nu \sin\theta_{\nu+1}) &  1 < \nu < n \\
			2F(0) - \Re F^{(2)}(0) (\sin\theta_{n-1} \sin\theta_{n} + \cos\alpha \sin\theta_{n} \sin\theta_{1}) &  \nu= n
		\end{cases}.
	$$
	In the limit $\alpha\to0$, we simply get
	$$
		\lim_{\alpha\to 0}{p(\nu)} =
			2F(0) - (\Re F^{(2)}(0)) (\sin\theta_{\nu-1}\sin\theta_\nu + \sin\theta_\nu \sin\theta_{\nu+1}).
	$$
\end{example}

\subsection{A Large sample with disorder}

We wish to consider a large ring with the coin matrix $C_\nu$ at each vertex~$\nu$ independently sampled from a common distribution. To this end, we consider a probability measure~$\mu$ on~$[0,2\pi]$ and introduce the product measure $(\mu\times\mu)^{\times \zz}$ on the product space~$\Omega := ([0,2\pi] \times [0,2\pi])^{\zz}$. We denote elements of~$\Omega$ in the form $\omega = (\omega^+_\nu, \omega^-_\nu)_{\nu\in\zz}$. We also fix real numbers~$t$ and $r$ with $tr \neq 0$ and $t^2 + r^2 = 1$.

Then, according to a random element~$\omega$ we set $$C_\nu(\omega) = \begin{pmatrix} \Exp{-\ii\omega_\nu^+} t & -\Exp{-\ii\omega_\nu^+} r \\ \Exp{-\ii\omega_\nu^-} r & \Exp{-\ii\omega_\nu^-} t \end{pmatrix}$$  and the unitary
\[
	\mathcal{W}(\omega) = \sum_{\nu \in \zz} \sum_{\tau = \pm 1}\delta_{\nu + \tau} \braket{\delta_\nu,{\cdot\,}} \otimes e_\tau \braket{e_\tau, C_\nu(\omega) {\,\cdot\,}} .
\]
on $\ell^2(\zz) \otimes \cc^2$. This is the form of quantum walk discussed in~\cite{JM10} (also see~\cite{Joy04,ASW}); it arises as the general form (up to unitary equivalence) of a disordered quantum walk on~$\zz$ where the quantum amplitudes of the transitions the right and to the left are independent random variables and the quantum transition probabilities between neighbouring sites are deterministic and independent of the site.

Note that
\begin{equation}\label{eq:covar}
	(S_\sys \otimes \one)\mathcal{W}(\omega)(S_\sys^* \otimes \one) = \mathcal{W}(\phi\omega)
\end{equation}
where $\phi$ is the shift $(\omega^+_\nu, \omega^-_\nu)_{\nu\in\zz} \mapsto (\omega^+_{\nu+1}, \omega^-_{\nu+1})_{\nu\in\zz}$ on~$\Omega$ and~$S_\sys^*$ is the periodic shift on~$\cc^n$.

The corresponding walk on a ring of $n$ sites has one-particle dynamics prescribed by the unitary
\begin{align*}
	W^{(n)}(\omega) &= \delta_{n-1} \braket{\delta_0,{\cdot\,}} e_{-1}\braket{e_{-1},C_0(\omega) {\,\cdot\,}}  + \delta_{1} \braket{\delta_0,{\cdot\,}} e_{+1}\braket{e_{+1},C_0(\omega) {\,\cdot\,}} \\
		& \qquad +
	\sum_{\nu=1}^{n-2} \sum_{\tau = \pm 1}\delta_{\nu + \tau} \braket{\delta_\nu,{\cdot\,}} \otimes e_\tau \braket{e_\tau, C_\nu(\omega) {\,\cdot\,}} \\
		&\qquad +
		\delta_{n-2} \braket{\delta_{n-1},{\cdot\,}} e_{-1}\braket{e_{-1},C_{n-1}(\omega) {\,\cdot\,}}
    \delta_{0} \braket{\delta_{n-1},{\cdot\,}} e_{+1}\braket{e_{+1},C_{n-1}(\omega) {\,\cdot\,}}.
\end{align*}
 on~$\ell^2(\{0,1,\dotsc,n-1\})\otimes\cc^2$. If we chose the state~$\psi_* = \delta_0 \otimes e_{-1}$ for the coupling with the environment~$\env$, we have to consider the random contraction~$\MM^{(n)}(\omega)$ defined by $\MM^{(n)} = (\one + (\cos\alpha-1)(\delta_0 \otimes e_{-1})\braket{\delta_0 \otimes e_{-1},{\cdot\,}})W^{(n)}$.

Let us suppose for simplicity that the support of~$\mu$ is a small nondegenerate interval. Let us also suppose that the series~\eqref{eq:sigma-plus} terminates after finitely many terms. Adapting slightly the arguments of~\cite{Joy04} and~\cite{ASW}, we see that since $\mathcal{W}$ is a band unitary matrix satisfying~\eqref{eq:covar} and since $\MM^{(n)}$ is a rank-six perturbation of $(\one_{[0,n-1]} \otimes \one) \mathcal{W} (\one_{[0,n-1]} \otimes \one)$, it follows from Birkhoff's ergodic theorem that
\[
	\lim_{n \to \infty} \frac{1}{2n} \tr(f (\MM^{(n)}(\omega))) = \lim_{n\to\infty}\frac{1}{2n}\tr((\one_{[0,n-1]} \otimes \one) f (\mathcal{W}(\omega)) (\one_{[0,n-1]} \otimes \one))
\]
for any polynomial~$f $ and for $\mu^{\times \zz}$-almost all~$\omega \in \Omega$. The right-hand side is usually written as the integral
\[
 \int_{\ci} f (\Exp{\ii\theta}) \d k(\theta)
\]
of~$f$ against the density of state~$k$ for~$\mathcal{W}$, defined through the Riesz--Markov representation theorem and which is an almost sure quantity in the sense that it is the same for $\mu^{\times \zz}$-almost all~$\omega \in \Omega$. In particular, this gives us, for large~$n$, an approximation of the asymptotic (in time) averaged (over the vertices of the graph) particle number density:
\[
	\lim_{t \to\infty} \frac{1}{n} \sum_{\nu=0}^{n-1} p_t(\nu) = \frac{1}{n} \tr(2 \Re \fonc (\MM^{(n)})) = 2 \int_{\ci} 2 \Re\fonc (\Exp{\ii\theta}) \d k(\theta) + o(1)
\]
as $n \to \infty$.

The techniques of~\cite{Joy04,ABJ} give us detailed information on the support of the density of states~$k$ for~$\mathcal{W}$. Indeed, the spectrum of the operator for $\omega$ identically 0 is made of the two bands
\[
	\Lambda_\pm = \{ x \pm \ii\sqrt{1-x^2}: x \in [-|t|,|t|]\}.
\]
Hence, a standard perturbation argument yields that for~$\mu$ supported on a small enough interval,
\[
	\operatorname{sp}(\mathcal{W}(\omega)) \subseteq \Big(\bigcup_{\theta \in \supp\mu} \Exp{\ii\theta} \Lambda_+\Big) \cup \Big(\bigcup_{\theta \in \supp\mu} \Exp{\ii\theta} \Lambda_-\Big)
\]
almost surely.

Therefore, by tuning $t$ and $\mu$ and taking $n$ large enough, one can bring the asymptotic average density $\lim_{t \to\infty} \frac{1}{n} \sum_{\nu=0}^{n-1} p_t(\nu)$ arbitrarily close to any value in the essential range of the function~$\ci \ni z \mapsto 2\Re F(z) \in \rr$.

\appendix
\section{Comments on the statistics}\label{sec:stats}

Following~\cite{HJ17}, we have made the choice of considering different species of fermions for the sample~$\sys$ and the environment~$\env$.
Considering the same species for both components of the system would have amounted to imposing the anticommutation relation~$\{a^\sharp(\psi),\tilde b(\varphi)\} = 0$ for all~$\psi \in \h_\sys$ and $\varphi \in \h_\env$ instead of the commutation relation~$[a^\sharp(\psi),b(\varphi)] = 0$.
This is realized on the Fock space~$\Ga(\h_\env) \otimes \Ga(\h_\sys)$ by setting~$\tilde b(\varphi) := b(\varphi) \otimes (-1)^{\dG(\one)}$.
In this case, one finds with the same techniques formulae such as
\[
	\tilde K_\alpha^* a^*(\psi) \tilde K_\alpha = a^*((\one + (\cos\alpha - 1)P)\psi) +\ii\sin\alpha \, \tilde b^*(\iota\psi),
\]
leading to
the same formulae as in Lemma~\ref{lem:one-tau}. Therefore, the asymptotics of the state in the sample~$\sys$ and the fluxes are the same.

With this choice of statistics, one may alternatively view~$\tilde K_\alpha$ as arising from the second quantization of a one-body operator on~$\h_\env \oplus \h_\sys$:
\[
	\tilde K_\alpha
    = \mathcal{U}\G(\one + (\cos\alpha-1)(\iota^*\iota + \iota\iota^*) - \ii\sin\alpha\,(\iota^* + \iota)) \mathcal{U}^*,
\]
where $\mathcal{U} : \Ga(\h_\env \oplus \h_\sys) \to \Ga(\h_\env) \otimes \Ga(\h_\sys)$ is the usual fermionic exponential map; see for example~\cite[\S{5.1}]{AJPP}.
The dynamics implemented by the unitary $\G((S \otimes U \oplus W)\Exp{-\ii\alpha(\iota+\iota^*)})$ gives rise to a quasi-free dynamics and the corresponding one-particle M{\o}ller operator
\[
	\Omega_+ = \slim\limits_{t \to \infty}(S \otimes U \oplus W)^{t}((S \otimes U \oplus W)\Exp{-\ii\alpha(\iota+\iota^*)})^{-t}
\]
exists and satisfies
\[
	\Omega_+ (0 \oplus \one) = \ii\sin\alpha \sum_{t'=0}^{\infty} (S\otimes U)^{t'+1}\iota W^* ((\one + (\cos\alpha-1)\iota^*\iota W^*)^{t'}.
\]
In particular, one quickly recovers
\[
	(0 \oplus \one) \Omega_+^* (\Sigma \oplus \Xi) \Omega_+ (0 \oplus \one) = 0 \oplus \Delta
\]
for all $\Xi \in \B(\h_\sys)$, showing that\,---\,at least when the initial state in the sample is a \textsc{giqf} state associated to a density~$\Xi$ invariant for the free dynamics\,---\,the limiting state is the same as in the case previously considered. This reduction to a one-body problem also suggests the same behaviour for Bose statistics.

\bibliographystyle{amsalpha}
\bibliography{walk-bib}

\end{document}